\documentclass[12pt]{article}

\usepackage{amssymb,a4wide}
\usepackage{upgreek}

\usepackage{epsfig}

\usepackage{xspace}

\usepackage{graphicx}
\usepackage{amsmath}
\usepackage{amsfonts}
\usepackage{amssymb}

\newcommand{\Xomit}[1]{ }

\newtheorem{theorem}{Theorem}

\newtheorem{lemma}[theorem]{Lemma}

\newtheorem{proposition}[theorem]{Proposition}

\newenvironment{proof}[1][Proof]{\textbf{#1.} }{\ \rule{0.5em}{0.5em}}

\date{}

\begin{document}


\title{An algorithm for the weighted metric dimension \\ of two-dimensional grids}

\author{Ron Adar\thanks{Department of Computer Science, University of Haifa, Haifa,
Israel. \texttt{radar03@csweb.haifa.ac.il.}} \and Leah
Epstein\thanks{Department of Mathematics, University of Haifa,
Haifa, Israel. \texttt{lea@math.haifa.ac.il.}}}

\date{}

\maketitle


\begin{abstract}
A two-dimensional grid consists of vertices of the form $(i,j)$
for $1\leq i \leq m$ and $1\leq j \leq n$, for fixed $m,n\geq 2$.
Two vertices are adjacent if the $\ell_1$ distance between their
vectors is equal to $1$. A landmark set is a subset of vertices $L
\subseteq V$, such that for any distinct pair of vertices $u,v \in
V$, there exists a vertex of $L$ whose distances to $u$ and $v$
are not equal. We design an efficient algorithm for finding a
minimum landmark set with respect to total cost in a grid graph
with non-negative costs defined on the vertices.
\end{abstract}

\date{}

\maketitle

\section{Introduction}
Consider an undirected graph $G=(V,E)$. For $u,v \in V$, let
$d(u,v)$ denote the edge distance between these two vertices. A
vertex $x\in V$ separates $u$ and $v$ if $d(x,u) \neq d(x,v)$, and
in this case, $x$ is also called a separating vertex for $u$ and
$v$. A landmark set is a subset $L \subseteq V$ such that for any
pair of vertices $u \neq v$, $L$ has at least one vertex $y$ that
separates $u$ and $v$. The vertices of a landmark set $L$ are often referred to as {\it landmarks}.
In the algorithmic metric dimension
problem, the goal is to find a landmark set $L$ of minimum
cardinality. In the weighted version of this problem, a
non-negative cost (or weight) function $c:V\rightarrow \mathbb
Q^+$ is given. For $U \subseteq V$, the cost or weight of $U$ is
defined as $c(U)=\sum_{a \in U} c(a)$, and the goal is to find a
landmark set $L$ minimizing $c(L)$. The cardinality of a minimum
cardinality landmark set of $G$ is called the {\it metric
dimension} of $G$, and the cost of a minimum cost landmark set is
called the {\it weighted metric dimension} of $G$.

A two-dimensional grid with integer parameters $m$ and $n$ has
$|V|=m\cdot n$ vertices of the form $(i,j)$, where $1 \leq i \leq
m$ and $1\leq j \leq n$. For vertices $(i_1,j_1)$, $(i_2,j_2)$,
let $((i_1,j_1),(i_2,j_2)) \in E $ if (and only if)
$|i_1-i_2|+|j_1-j_2|=1$. The resulting distance between two
vertices is the $\ell_1$ distance between their vectors, that is,
$d((i_1,j_1),(i_2,j_2))=|i_1-i_2|+|j_1-j_2|$. This graph can be
visualized on the plane, such that the rows are numbered from top
to bottom, and its columns from left to right. The sides of the
grid are the top row (row $1$), the bottom row (row $m$), the
leftmost column (column $1$), and the rightmost column (column
$n$). The $j$th vertex in the $i$th row of the grid is denoted by
$(i,j)$. The vertices $(1,1)$, $(1,n)$, $(m,1)$, and $(m,n)$ are
called corners. That is, vertices of degree $2$ are corners, and
vertices of degrees below $4$ belong to sides. Other vertices (of
degree $4$) are called internal. Since a minimum cardinality
landmark set consists of  a single vertex if and only if the graph
is a path  \cite{KRR1996}, and the case of a path (a
one-dimensional grid) was completely studied \cite{KRR1996,ELW},
we assume that $m\geq 2$ and $n\geq 2$, and therefore any landmark
set will have at least two vertices. Corners that belong to the
same row or to the same column are called adjacent corners, and
otherwise they are non-adjacent or opposite corners. Sides that
share a corner are called adjacent sides, and otherwise they are
non-adjacent or opposite sides. We assume that the vertex costs
are given in a matrix of size $\Theta(mn)=\Theta(|V|)$, such that
any specific cost (the value $c(v)$ for a given vertex $v$) can be
retrieved in time $O(1)$. We let $c_{i,j}=c((i,j))$ for any $1\leq
i \leq m$ and $1\ leq j \leq n$. Let a double side consist of two
adjacent sides, excluding their common corner.  We say that two
vertices $x_1=(y_1,z_1)$ and $x_2=(y_2,z_2)$ are a {\it on a joint
diagonal} if $y_1+z_1=y_2+z_2$. Such pairs of vertices are of
particular interest as we should be careful regarding separating
them, and in particular, the corner vertex $(1,1)$ does not
separate any such pair. For two vertices $r_1=(a_1,b_1)$ and
$r_2=(a_2,b_2)$ such that $a_1\leq a_2$ and $b_1 \leq b_2$, we
define the sub-grid of $r_1$ and $r_2$ as the set of all vertices
whose first component is in $[a_1,a_2]$ and their second component
is in $[b_1,b_2]$.

In this work, we will use the standard term {\it minimal} for a
landmark set that is minimal with respect to set inclusion. We
will use the term {\it minimum} landmark set for a landmark set
that is minimum with respect to cost. A minimum cardinality
landmark set will be called {\it smallest}. As weights are
non-negative, there always exists a minimum cost landmark set that
is also a minimal landmark set. In some cases, when we search for
a minimum landmark set, we will only consider minimal landmark
sets as potential solutions. We will show, in particular, that the
cardinality of a minimal landmark set is either a positive number
in $\{2,4,\ldots,2\cdot \min\{m,n\}-2\}$ (note that the upper
bound was shown in \cite{ACM14}), or it is equal to $3$. We find
that if $m=n=2$, all minimal landmark sets have cardinality $2$,
if $\min\{m,n\}=2$ but $\max\{m,n\}>2$, all minimal landmark sets
have cardinalities of $2$ and $3$. We will show that any minimal
landmark set of cardinality at least $4$ has a specific form, and
we use dynamic programming to find a subset of minimum cost of
this form (there can be sets of this form that are landmark sets
but they are not minimal landmark sets). Moreover, it follows from
our results that the case of cardinality $3$ is the only possible
case of an odd cardinality of a minimal landmark set. We also
analyze minimal landmark sets of cardinalities $2$ and $3$. The
result for cardinality $2$ was obtained by Melter and Tomescu
\cite{MT} (and generalized by Khuller, Raghavachari, and Rosenfeld
\cite{KRR1996}), where landmark sets of minimum cardinality were
studied. The result for cardinality $3$ was obtained in
\cite{ACM14}, where properties of some minimal landmark sets are
studied. For completeness, and as the proofs some of these
properties are used later as well, we provide complete proofs.
These proofs are followed by efficient algorithms for finding such
sets. Our main algorithm applies several algorithms and provides a
minimum landmark set out of landmark sets of cardinalities $2$,
$3$, and at least $4$. The output, which is a set of minimum cost
out of the outputs, is a minimum landmark set. Our main result is
therefore an efficient (polynomial-time) algorithm for finding a
minimum (i.e., minimum cost) landmark set in a two-dimensional
grid graph. That is, we solve the algorithmic weighted metric
dimension problem on two-dimensional grid graphs. The cases of
landmark sets of cardinalities $2$ and $3$ are relatively simple,
and the main technical difficulty is to find a minimum landmark
set out of landmark sets of cardinality at least $4$. We will
observe that every such set is related to a sequence that follows
a pattern, which we will call a zigzag sequence.

Another variant of grid graphs, where the distances are according
to the $\ell_{\infty}$ norm was studied \cite{KRR1996,ST04}. This
first articles on the metric dimension problem were by Harary and
Melter \cite{HM1976} and by Slater \cite{Slat75}. The problem is
NP-hard \cite{KRR1996} and hard to approximate \cite{BE+06,DPL11}
for general graphs, and it was studied for specific graph classes
\cite{HM1976,Slat75,KRR1996,CEJO00,Babai,SSH02,CZ03,CH+07,ELW}.
Applications can be found in
\cite{BE+06,HM1976,MT,Chvatal,KRR1996,CEJO00}, where some of these
applications are relevant for weighted graphs (see also
\cite{ELW}).

\section{Main result}
We start with proving some simple but crucial properties.

\begin{lemma}
Any landmark set has at least one vertex of each double side.
\end{lemma}
\begin{proof}
Without loss of generality consider the first row and the first
column. We show that no vertex separates vertices $(1,2)$ and
$(2,1)$ except for vertices of this double side. For any vertex
$(a,b)$ such that $a \geq 2$ and $b \geq 2$, we find
$d((1,2),(a,b))=a+b-3$ and $d((2,1),(a,b))=a+b-3$ (any such
shortest path traverses $(2,2)$). Moreover, $d((1,2),(1,1)=1$ and
$d((2,1),(1,1)=1$. The remaining vertices are on the double side,
where any such vertex has either $a=1$ and $b>1$ or it has $a>1$
and $b=1$. If $a=1$ and $b\neq 1$, then $d((1,2),(a,b))=b-2$ and
$d((2,1),(a,b))=b$, and if $a \neq 1$ and $b=1$, then
$d((1,2),(a,b))=a$ and $d((2,1),(a,b))=a-2$. Therefore given a
landmark set $L$, at least one vertex of the double side must
belong to $L$.
\end{proof}

\begin{lemma}\label{opC}
No minimal landmark set contains two opposite corners.
\end{lemma}
\begin{proof}
Without loss of generality, consider the corners $(1,1)$ and $(m,n)$. For any $x=(y,z)$, $d(x,(1,1))=y+z-2$ and $d(x,(m,n))=m+n-y-z$. Thus, for any two vertices, their distances to $(1,1)$ are distinct if and only if their distances to $(m,n)$ are distinct.
\end{proof}

\begin{lemma}\label{opSides}
For any landmark set, there is a pair of opposite sides of the
grid such that each one of these sides has a landmark.
\end{lemma}
\begin{proof}
If every side has a landmark, we are done. Otherwise, consider a
side $\Lambda$ without a landmark. Since every double side has a
landmark, each one of the two sides adjacent to $\Lambda$ has a
landmark (and these are two distinct landmark as the two sides are
disjoint).
\end{proof}

As mentioned above, the following was proved in \cite{MT}.

\begin{proposition}
A set that consists of exactly two vertices is a landmark set if
and only if these two vertices are adjacent corners.
\end{proposition}
\begin{proof}
First, note that a landmark set of cardinality $2$ must be minimal
as any landmark set for a graph that is not a path has cardinality
of at least $2$ \cite{KRR1996}.

Consider a set $A=\{v_1=(a_1,b_1),v_2=(a_2,b_2)\}$, where either $a_1\neq a_2$ or $b_1\neq b_2$ or both.

First, assume that $v_1$ and $v_2$ are adjacent corners, and
without loss of generality, $a_1=a_2=1$, $b_1=1$, and $b_2=n$.
Consider two distinct vertices $x_1=(y_1,z_1)$ and
$x_2=(y_2,z_2)$. For $i=1,2$, we have
$d(x_i,v_1)=|y_i-a_1|+|z_i-b_1|=y_i+z_i-2$. If $x_1$ and $x_2$ are
not on a joint diagonal, $y_1+z_1 \neq y_2+z_2$, and we have
$d(x_1,v_1)\neq d(x_2,v_1)$, so $v_1$ separates them. If $x_1$ and
$x_2$ are on a joint diagonal, then for $i=1,2$, we have
$d(x_i,v_2)=|y_i-a_2|+|z_i-b_2|=y_i-1+m-z_i$. Since $y_1+z_1 =
y_2+z_2$, we have $d(x_1,v_2) = y_1-z_1+m-1=y_2+z_2-2z_1+m-1$
while $d(x_2,v_2)=y_2-z_2+m-1$. If $d(x_1,v_2)=d(x_2,v_2)$, we get
$z_1=z_2$, and therefore by $y_1+z_1 = y_2+z_2$, we also find
$y_1=y_2$, proving $x_1=x_2$. Thus, if $x_1\neq x_2$, at least one
of $v_1$ or $v_2$ separates them. This shows that $A$ is a
landmark set.

If $v_1$ and $v_2$ are opposite corners, then by Lemma \ref{opC}
cannot be a minimal landmark set. Next, assume that at least one
of $v_1$ and $v_2$ is not a corner. Assume without loss of
generality that $A$ does not contain any corner, except for
possibly $(1,1)$. If $A$ contains a corner, the other vertex of
$A$ is a vertex of the double side consisting of the first row and
the first column. This last vertex is not a corner by the
assumption that $(1,n),(m,1),(m,n)\notin A$ (and since $(1,1)$
does not belong to this double side). In this case the landmark
set has no vertex of the double side consisting of the last row
and the last column, contradicting the property that it is a
landmark set.

If $A$ does not contain any corner, then since any landmark set
has a pair of vertices on opposite sides, its two vertices are on
opposite sides. Assume without loss of generality (due to
symmetry) that $A=\{(1,z),(m,z')\}$, where $2 \leq z \leq z' \leq
m-1$. If $z=z'$, then $d((1,z-1),(1,z))=1$, $d((1,z+1),(1,z))=1$,
$d((1,z-1),(m,z))=m$, and $d((1,z+1),(m,z))=m$, so no vertex of
$A$ separates $(1,z-1)$ and $(1,z+1)$. Otherwise,
$d((1,z+1),(1,z))=1$, $d((2,z),(1,z))=1$,
$d((1,z+1),(m,z'))=m+z'-z-2$, and $d((2,z),(m,z'))=m+z'-z-2$, so
no vertex of $A$ separates $(1,z+1)$ and $(2,z)$.
\end{proof}

\begin{lemma}\label{ab}
Let $z$ satisfy $1 \leq z < n$, and let $(a,b)$ be a vertex that
separates the vertices $(1,z+1)$ and $(2,z)$. If $b \leq z$, then
$a>1$ and if $b \geq z+1$, then $a=1$.
\end{lemma}
\begin{proof}
If $b \leq z$ and $a=1$, then
 $d((1,z+1),(a,b))=d((2,z),(a,b))=1+z-b$. This is a contradiction to the role of $(a,b)$ as a separating vertex for
$(1,z+1)$ and $(2,z)$, and therefore in the case $b \leq z$, we
have $a>1$. Otherwise, assume that $b\geq z+1$ holds. We have
$d((1,z+1),(a,b))=a-1+b-z-1=a+b-z-2$, $d((2,z),(a,b))=|a-2|+b-z$.
Thus, as $(a,b)$ separates these two vertices, $a=1$.
\end{proof}

In the next lemma we consider a sub-grid of two vertices
$(a_1,b_1)$, $(a_2,b_2)$ such that $a_1 \leq a_2$ and $b_1 \leq
b_2$, and a vertex $(a,b)$ that is either on the left hand side of
this sub-grid ($b\leq b_1$ and $a_1 < a \leq a_2$) or it is above
this sub-grid ($a\leq a_1$ and $b_1<b \leq b_2$). We also consider
the smaller sub-grid whose upper left corner is $(a,b_1)$ in the
first case and $(a_1,b)$ in the second case, and the other corner
remains $(a_2,b_2)$. We show that $(a,b)$ separates any pair of
vertices on a joint diagonal that are not both vertices of the
smaller sub-grid.

\begin{lemma}\label{cutting}
Let $(a_1,b_1)$, $(a_2,b_2)$ be  grid vertices  such that $a_1
\leq a_2$ and $b_1 \leq b_2$. Let $u=(a,b)$ be a vertex such that
either $a_1 < a \leq a_2$ and $b\leq b_1$ hold or $a \leq a_1$ and
$b_1<b\leq b_2$ hold. Then, $(a,b)$ separates any pair of distinct
vertices $v_1=(x_1,y_1)$ and $v_2=(x_2,y_2)$ that are on a joint
diagonal and $x_1<x_2$ (so $x_1+y_1=x_2+y_2$ and $y_1>y_2$ hold)
under the required conditions, where the conditions on $(a,b)$ are
as follows. In the first option for $(a,b)$, it holds that
$a_1\leq x_1 < a$, $x_1< x_2 \leq a_2$, and $b_1\leq y_2 < y_1
\leq b_2$, and in the second option for $(a,b)$, it holds that
$a_1\leq x_1 < x_2 \leq a_2$, $b_1\leq y_2 < b$, and $y_2< y_1
\leq b_2$.
\end{lemma}
\begin{proof}
If $a_1=a_2$ or $b_1=b_2$, there are no such pairs $v_1,v_2$. Thus
we assume $a_1<a_2$ and $b_1<b_2$. Since the two options for
$(a,b)$ are analogous, we will prove the property for the first
option.

We have $|x_1-a|=a-x_1$, $|y_1-b|=y_1-b$, and $|y_2-b|=y_2-b$.
Thus, $d(v_1,u)=y_1-x_1+a-b$ and $d(v_2,u)=|x_2-a|+y_2-b$. Assume
by contradiction that $d(v_1,u)=d(v_2,u)$. We get
$|x_2-a|=y_1-x_1-y_2+a$, and by using $x_1+y_1=x_2+y_2$, we have
$|x_2-a|=x_2-2x_1+a$. If $x_2 \leq a$, this implies $x_1=x_2$, a
contradiction.  If $x_2 > a$, this implies $a=x_1$, a
contradiction as well.
\end{proof}

Obviously, in the case $a = a_2$ and $b\leq b_1$, and in the case
$a \leq a_1$ and $b = b_2$, the lemma shows that $(a,b)$ separates
any pair of vertices on a joint diagonal of the sub-grid of
$(a_1,b_1)$ and $(a_2,b_2)$.

In the next lemma we show that it is possible that while every
pair of vertices should be separated by any landmark set, it is
possible to restrict the set of pairs that should be tested. More
precisely, given two landmarks on opposite sides (we consider the
case of the top row and the bottom row, such that the vertex of
the top row is strictly to the left of the vertex of the bottom
row), creating a sub-grid, it will be sufficient to ensure for
every pair of vertices on a joint diagonal, both being vertices of
the sub-grid, are separated.

\begin{lemma}\label{onlydiag}
Consider $X \subseteq V$, where $X$ contains two side vertices (of
opposite sides) $t_1=(1,z_1)$, $t_2=(m,z_2)$ such that $1\leq
z_1<z_2 \leq n$. If for any pair of vertices on a joint diagonal
of the sub-grid of $t_1$ and $t_2$, $(a_1,b_1)$ and $(a_2,b_2)$
such that $a_1<a_2$ (so $a_1+b_1=a_2+b_2$ and $z_1\leq b_2<b_1\leq
z_2$), $X$ has a vertex that separates $(a_1,b_1)$ and
$(a_2,b_2)$, then $X$ is a landmark set.
\end{lemma}
\begin{proof}
If $X$ contains a vertex of the form $(1,z')$ such that
$z_1<z'<z_2$, it is sufficient to prove the claim for $(1,z')$ and
$(m,z_2)$ (and this will imply the claim for $(1,z_1)$ and
$(m,z_2)$). Thus, without loss of generality we will assume that
no such vertex belongs to $X$. Similarly, we assume that no vertex
of the form $(m,z')$ with $z_1<z'<z_2$ belongs to $X$.

Consider the vertices $(1,z_1+1)$ and $(2,z_1)$. These vertices
are on a joint diagonal and they are vertices of the considered
sub-grid, and thus by the conditions of the lemma, $X$ has a
vertex that separates them. Let this vertex be $(a,b)$. By Lemma
\ref{ab}, none of $t_1$ and $t_2$ separates these two vertices,
and thus $(a,b)$ is another vertex satisfying $a=1$ and $b\geq
z_2$ or $a>1$ and $b\leq z_1$ (the case where $z_1+1 \leq b \leq
z_1-1$ is impossible since in this case $a=1$ and we assume that
no vertex $(1,z')$ with $z_1<z'<z_2$ belongs to $X$).

Let $v_1=(x_1,y_1)$ and $v_2=(x_2,y_2)$, such that $y_1 \leq y_2$
be a pair of distinct vertices. Assume that they are not separated
by $(1,z_1)$,  by $(m,z_2)$, or by $(a,b)$. That is, we assume
$d(v_1,(1,z_1))=d(v_2,(1,z_1))$,  $d(v_1,(m,z_2))=d(v_1,(m,z_2))$
and $d(v_1,(a,b))=d(v_2,(a,b))$. We find
$d(v_i,(1,z_1))=x_i-1+|y_i-z_1|$ and
$d(v_i,(m,z_2))=m-x_i+|y_i-z_2|$. If $(a,b)$ is such that $a=1$
and $b \geq z_2$, then $d(v_i,(a,b))=x_i-1+|y_i-b|$, and otherwise
$d(v_i,(a,b))=|x_i-a|+|y_i-b|$.

We consider all possible cases with respect to the columns of
$v_1$ and $v_2$.

\begin{figure} [h!]
\hspace{0.8in}
\includegraphics[angle=0,width=0.7\textwidth]{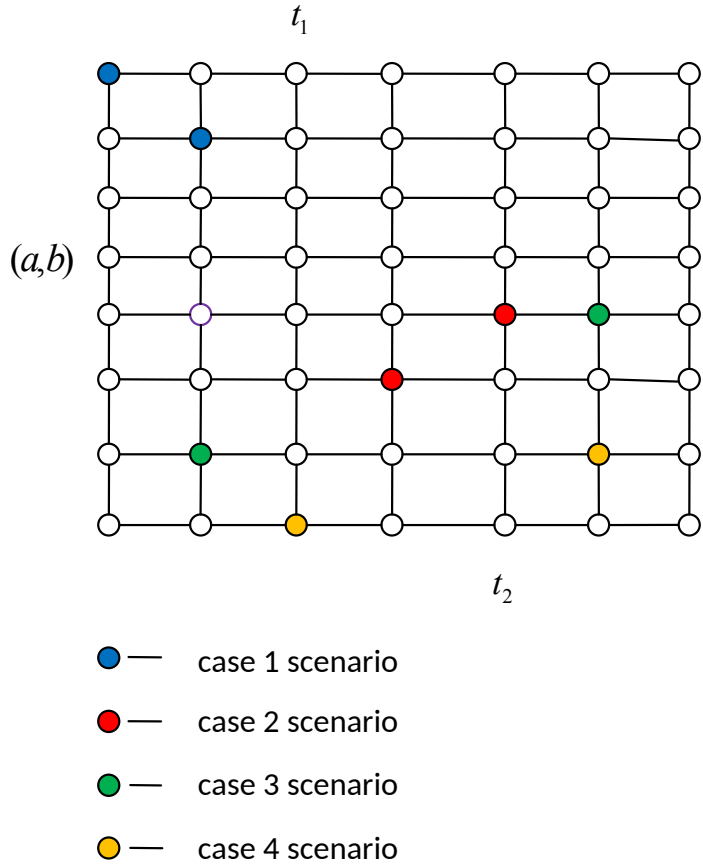}
\caption{An example of a grid with $8$ rows and $7$ columns. All possible cases of $v_1$ and $v_2$ described in the proof of Lemma \ref{onlydiag}, for the sub-grid of $t_1=(1,3)$ and $t_2=(8,5)$, where $(a,b)=(4,1)$ is a vertex separating (for example) the pair $(2,3)$,$(1,4)$ and the pair $(5,3)$,$(3,5)$, but it does not separate the pair $(5,3)$,$(4,4)$. \label{casesexmp}} 
\end{figure}

{\noindent{\bf Case 1. \ }} In this case either $y_1\leq y_2 <z_1$
or $z_2<y_1\leq y_2$ holds. That is, both $v_1$ and $v_2$ are not
vertices of the sub-grid, and they are on the same side of the
sub-grid (either to the left or to the right of it, see for example the blue vertices in figure \ref{casesexmp}).

If $y_2 < z_1$, we have $d(v_i,(1,z_1))=x_i-1+z_1-y_i$ and
$d(v_i,(m,z_2))=m-x_i+z_2-y_i$. We find $x_1-y_1=x_2-y_2$ and
$x_1+y_1=x_2+y_2$. Similarly, if $y_1
> z_2$, we have $d(v_i,(1,z_1))=x_i-1+y_i-z_1$ and
$d(v_i,(m,z_2))=m-x_i+y_i-z_2$, and in this case We find
$x_1-y_1=x_2-y_2$ and $x_1+y_1=x_2+y_2$ hold as well. In both
cases we find that $x_1=x_2$ and $y_1=y_2$ hold, contradicting the
property that $v_1 \neq v_2$.

{\noindent{\bf Case 2. \ }} In this case  $z_1 \leq y_1 \leq y_2
\leq z_2$ holds. That is, both vertices are vertices of the
sub-grid (see for example the red vertices in figure \ref{casesexmp}).

We have $d(v_i,(1,z_1))=x_i-1+y_i-z_1$,  and
therefore $x_1+y_1=x_2+y_2$ holds. In this case $v_1$ and $v_2$
are on a joint diagonal of the sub-grid, and by assumption there
is a vertex of $X$ separating them.

{\noindent{\bf Case 3. \ }} In this case $y_1 < z_1$ and $y_2 >
z_2$, that is, both $v_1$ and $v_2$ are not vertices of the
sub-grid, one of them ($v_1$) is to the left of the sub-grid and
the other one ($v_2$) is on the right (see for example the green vertices in figure \ref{casesexmp}).

We have $d(v_1,(1,z_1))=x_1-1+z_1-y_1$,
$d(v_2,(1,z_1))=x_2-1+y_2-z_1$, $d(v_1,(m,z_2))=m-x_1+z_2-y_1$,
and $d(v_2,(m,z_2))=m-x_2+y_2-z_2$. This proves
$x_1+z_1-y_1=x_2+y_2-z_1$ and $x_1+y_1-z_2=x_2-y_2+z_2$. Taking
the sum and difference we get $x_1-x_2=z_2-z_1>0$ (so $x_1>x_2$) and
$y_1+y_2=z_1+z_2$.

Let $(a',b') \in X$ be a vertex that separates the vertices
$(x_2,z_1+1)$ and $(x_2+1,z_1)$. These are vertices of $G$ since
$x_2+1\leq x_1$ and $z_1+1\leq z_2$. These two vertices are
vertices of the sub-grid on a joint diagonal, and therefore
$(a',b')$ exists according to the conditions of the lemma. First, we analyze
the values of $a'$ and $b'$. If $a' \leq x_2$ and $b' \leq z_1$, we
find $d((x_2,z_1+1),(a',b'))=x_2-a'+z_1+1-b'$ and
$d((x_2+1,z_1),(a',b'))=x_2+1-a'+z_1-b'$, so $(a',b')$ does not
separate the two vertices. If $a' \geq x_2+1$ and $b' \geq z_1+1$,
we find $d((x_2,z_1+1),(a',b')=a'-x_2+b'-z_1-1$ and
$d((x_2+1,z_1),(a',b'))=a'-x_2-1+b'-z_1$, so $(a',b')$ does not
separate the two vertices. Thus either $a' \leq x_2$ and $b' \geq
z_1+1$ hold or $a' \geq x_2+1$ and $b' \leq z_1$ hold.

We consider the two cases. In the
first case, by $x_1>x_2\geq a'$ and $y_1<z_1<b'$, we have
$d(v_1,(a',b'))=x_1-a'+b'-y_1$. Moreover,
$d(v_2,(a',b'))=x_2-a'+|y_2-b'|$. These two values are distinct
since if $y_2 \geq b'$, then
$x_1-x_2+b'-y_1-|y_2-b'|=x_1-x_2-y_1-y_2+2b'=(z_2-z_1)-(z_1+z_2)+2b'=2(b'-z_1)>0$,
and if $y_2 \leq b'$, then
$x_1-x_2+b'-y_1-|y_2-b'|=x_1-x_2-y_1+y_2=(z_2-z_1)+y_2-(z_1+z_2-y_2)=2(y_2-z_1)>0$, since $y_2>z_2>z_1>y_1$.

In the second case, by $x_2\leq a'-1<a'$ and $b' \leq z_1<y_2$, we have $d(v_2,(a',b'))=a'-x_2+y_2-b'$.
We show that the two values $d(v_1,(a',b'))$ and $d(v_2,(a',b'))$ are distinct.

If $y_1\leq b'$ and $x_1\leq a'$, $d(v_1,(a',b'))=a'-x_1+b'-y_1$. The two distances are distinct since we have $x_1+y_1+y_2-x_2-2b'=z_2-z_1+z_1+z_2-2b'=2(z_2-b')>0$, since $z_2>z_1\geq b'$.

If $y_1\leq b'$ and $x_1\geq a'+1$, $d(v_1,(a',b'))=x_1-a'+b'-y_1$. The two distances are distinct since we have $y_1-x_1+y_2-x_2+2a'-2b'=(z_1+z_2)-2x_2+z_1-z_2+2a'-2b'=2(z_1-x_2+a'-b')>0$ since $a' \geq x_2 + 1$ and $z_1 \geq b'$.

If $y_1\geq b'+1$ and $x_1\leq a'$, $d(v_1,(a',b'))=a'-x_1+y_1-b'$. The two distances are distinct since we have $-y_1+x_1+y_2-x_2=z_2-z_1+2y_2-z_1-z_2=2(y_2-z_1)>0$ (by $y_2>z_2>z_1$).

If $y_1\geq b'+1$ and $x_1\geq a'+1$, $d(v_1,(a',b'))=x_1-a'+y_1-b'$. The two distances are distinct since we have $-y_1-x_1+2a'+y_2-x_2=z_1+z_2-2y_1+2a'-2x_2+z_2-z_1=2(z_2-y_1+a'-x_2)>0$ since $a' \geq x_2+1$ and $z_2>y_1$.

{\noindent{\bf Case 4. \ }} In this case we either have $y_1 <
z_1$ and $z_1 \leq y_2 \leq z_2$ or we have $z_1 \leq y_1 \leq
z_2$ and $y_2 > z_2$. That is, one vertex is a vertex of the
sub-grid, while the other one is not a vertex of the sub-grid
(see for example the yellow vertices in figure \ref{casesexmp}).

In the first option, we have $d(v_1,(1,z_1))=x_1-1+z_1-y_1$,
$d(v_2,(1,z_1))=x_2-1+y_2-z_1$, $d(v_1,(m,z_2))=m-x_1+z_2-y_1$,
and $d(v_2,(m,z_2))=m-x_2+z_2-y_2$. This proves
$x_1+z_1-y_1=x_2+y_2-z_1$ and $x_1+y_1=x_2+y_2$. The two
properties together imply $y_1=z_1$, which is a contradiction. In
the second option, we have $d(v_1,(1,z_1))=x_1-1+y_1-z_1$,
$d(v_2,(1,z_1))=x_2-1+y_2-z_1$, $d(v_1,(m,z_2))=m-x_1+z_2-y_1$,
and $d(v_2,(m,z_2))=m-x_2+y_2-z_2$. This proves $x_1+y_1=x_2+y_2$
and $x_1+y_1-z_2=x_2-y_2+z_2$. The two properties together imply
$y_2=z_2$, which is a contradiction.
\end{proof}

When we say that a property holds up to rotation of  the grid or
mirroring it, we mean that we consider the same grid graph but the
numbering of rows and columns is different. In particular, there
are four choices for which corner is $(1,1)$, and given that
choice, there are two choices regarding which one of its two
neighbors is denoted by $(1,2)$ (and which one is denoted by
$(2,1)$). Fixing $(1,1)$ and $(1,2)$, the numbering of the other
vertices is unique. Thus, there are eight ways to number the
vertices.

In the following analysis, we will assume that the top row and
bottom row are two opposite sides that contain landmarks
(otherwise, if one of these sides does not contain a landmark, we
can rotate the grid). Obviously, a set may contain more than two
vertices on two opposite sides. Given a set of vertices, out of
pairs of vertices such that one is on the top row and the other is
on the bottom row, we will always select two vertices such that
the absolute value of the difference between the indices of their
columns is minimal. By possibly mirroring the grid, given two such
vertices, we will assume that the index of the column of the
vertex of the bottom row is not smaller of the index of the column
of the vertex of the top row. Thus, any subset of vertices which
we will discuss has two vertices $(1,z)$ and $(m,z')$, where
$z'\geq z$. It is obviously possible that a landmark set will
contain additional vertices of these two sides, and it may contain
vertices of other sides, and internal vertices. By the choice of
these two vertices from a given subset of vertices (such that $|z'-z|$ is minimal), no vertex
$(1,z'')$ such that $z<z''<z'$ is an element of the set and no
vertex $(m,z'')$ such that $z<z''<z'$ is an element of the set.
Moreover, if $z\neq z'$, $(1,z')$ and $(m,z)$ are also not
elements of this set. Since every landmark set has such a pair of
landmarks on opposite sides (by Lemma \ref{opSides}), in what follows we only consider sets that contain this pair
of vertices. We will assume that $1 < z<z'\leq m$ or $1\leq z<z'<m$ holds (that is, at most one of $(1,z)$ and $(m,z')$ is a corner), as a minimal landmark set does not contain a pair of opposite corners. Moreover, in the case where $z=z'$ and either $z=1$ or $z=n$, these two vertices are adjacent corners, and a minimal landmark set containing these two vertices has cardinality $2$. Thus, in the analysis of landmark sets of cardinality at least $3$, we assume that if $z=z'$, then $1<z<n$ holds.


\begin{lemma}\label{lemab}
For a landmark set $L$ such that $(1,z), (m,z') \in L$, where $1\leq z<z'\leq m$, $L$ has at least one vertex $(a,b)$ such that either $b \leq z$ and $a>1$ hold or $b > z'$ and $a=1$ hold.
\end{lemma}
\begin{proof}
Consider a vertex $(a,b)$ that separates $(1,z+1)$ and $(2,z)$. By
Lemma \ref{ab}, none of $(1,z)$ and $(m,z')$ separates $(1,z+1)$ and $(2,z)$. Moreover, as by the choice of $(1,z)$ and $(m,z')$, no
vertex of the form $(1,z'')$ for $z < z'' \leq z'$ is in $L$, if $b \geq z+1$ holds, then the stronger condition $b
\geq z'+1$ holds as well.
\end{proof}

\begin{lemma}
Any set of the form $\{(1,z),(m,z'),(1,z'')\}$ with $1 \leq z < z' <z'' \leq n$ is a landmark set.
\end{lemma}
\begin{proof}
By Lemma \ref{onlydiag}, it is sufficient to prove that $(1,z'')$
separates any pair of vertices on a joint diagonal, and they are vertices of the sub-grid
of $(1,z)$ and $(m,z')$. Let $v_1=(x_1,y_1)$ and $v_2=(x_2,y_2)$
be such that $z \leq y_2<y_1 \leq z'$, and $x_1+y_1=x_2+y_2$ (so
$x_1 <x_2$). We have $d(v_i,(1,z''))=x_i-1+z''-y_i$. The two
distances are distinct as
$(x_2-1+z''-y_2)-(x_1-1+z''-y_1)=x_2-x_1+y_1-y_2=(x_1+y_1-y_2)-x_1+y_1-y_2=2(y_1-y_2)\neq
0$.
\end{proof}

By the last lemma and rotating the grid, any set of the form $\{(1,z),(m,z'),(m,z'')\}$ with $1 \leq z'' < z < z' \leq n$ is a landmark set as well.


\begin{lemma}\label{threetogether}
A minimal landmark set $L$ does not contain three vertices of one row or of one column.
\end{lemma}
\begin{proof}
We prove the claim for a column, the proof for a row is analogous. Consider three vertices $(a_1,b)$, $(a_2,b)$, and $(a_3,b)$, where $a_1<a_2<a_3$ and $1\leq b\leq n$. We show that every pair of vertices separated by $(a_2,b)$ is separated by at least one of the other two vertices. Assume that there exists a pair of vertices $v_1=(x_1,y_1)$ and $v_2=(x_2,y_2)$, where $v_1\neq v_2$ and $x_1\leq x_2$, that are not separated by $(a_1,b)$ or $(a_3,b)$. Thus, $|x_1-a_1|+|y_1-b|=|x_2-a_1|+|y_2-b|$ and $|x_1-a_3|+|y_1-b|=|x_2-a_3|+|y_2-b|$. Taking the difference between the inequalities, we get $|x_1-a_1|-|x_1-a_3|=|x_2-a_1|-|x_2-a_3|$.

If $x_1 \leq a_1 < x_2 < a_3$, we have $|x_1-a_1|-|x_1-a_3|-|x_2-a_1|+|x_2-a_3|=a_1-x_1-a_3+x_1-x_2+a_1+a_3-x_2=2(a_1-x_2)<0$, a contradiction.
If $x_1 \leq a_1$ and $a_3 \leq x_2$, we have $|x_1-a_1|-|x_1-a_3|-|x_2-a_1|+|x_2-a_3|=a_1-x_1-a_3+x_1-x_2+a_1+x_2-a_3=2(a_1-a_3)<0$, a contradiction as well. Analogously, we can prove that the case $a_1 < x_1 < a_3$ and $x_2\geq a_3$ leads to a contradiction. If $a_1<x_1<x_2<a_3$, then
$|x_1-a_1|-|x_1-a_3|-|x_2-a_1|+|x_2-a_3|=x_1-a_1-a_3+x_1-x_2+a_1+a_3-x_2=2(x_1-x_2)<0$, a contradiction.

Thus, one of $x_1,x_2\leq a_1$, $a_1<x_1=x_2<x_3$, or $x_1,x_2 \geq a_3$ holds.
We show that $d(v_1,(a_2,b))=d(v_2,(a_2,b))$, that is, $|x_1-a_2|+|y_1-b|=|x_2-a_2|+|y_2-b|$. To show this, it is sufficient to show that $|x_1-a_1|-|x_1-a_2|=|x_2-a_1|-|x_2-a_2|$ holds. The equality obviously holds if $x_1=x_2$.
If $x_1,x_2\leq a_1$, then $|x_1-a_1|-|x_1-a_2|-|x_2-a_1|+|x_2-a_2|=a_1-x_1-a_2+x_1-a_1+x_2+a_2-x_2=0$, and if $x_1,x_2 \geq a_3$, then $|x_1-a_1|-|x_1-a_2|-|x_2-a_1|+|x_2-a_2|=x_1-a_1-x_1+a_2-x_2+a_1+x_2-a_2=0$.
\end{proof}

\begin{lemma}\label{zzprime}
Consider a set of the form $Y=\{(1,z),(m,z),(a,b)\}$, such that $1 \leq z \leq n$. This set is a minimal landmark set if and only if $b \neq z$ and $z\neq 1,n$.
\end{lemma}
\begin{proof}
If $z=1$ or $z=n$, then $\{(1,z),(m,z)\}$ is a landmark set, and therefore $Y$ is not a minimal landmark set. Otherwise, assume $b=z$.
Consider the vertices $(1,z-1)$ and $(1,z+1)$. For any $1 \leq r
\leq n$, We have $d((1,z-1),(r,z))=r$ and $d((1,z+1),(r,z))=r$.
Thus, if $b=z$, these two vertices do not have a separating vertex
in the set $Y$.

We show that in the remaining cases $Y$ is indeed a minimal landmark set. We assume $b\neq z$ and $1<z<n$.
Since none of $(1,z)$ and $(m,z)$ is a corner, no proper subset of $Y$ is a landmark set of cardinality $2$. It is left to show that any pair of vertices is separated by a vertex of $Y$. Consider two vertices $v_1=(x_1,y_1)$ and $v_2=(x_2,y_2)$, where $y_1 \leq y_2$. Assume that $d(v_1,(1,z))=d(v_2,(1,z))$ and $d(v_1,(m,z))=d(v_2,(m,z))$ hold. We will show that $d(v_1,(a,b)) \neq d(v_2,(a,b))$.

We find $x_1-1+|y_1-z|=x_2-1+|y_2-z|$ and $m-x_1+|y_1-z|=m-x_2+|y_2-z|$. Taking the difference between the last two equalities we get $x_1=x_2$. Moreover, if $y_1,y_2 \leq z$ or $y_1, y_2 \geq z$, we also get $y_1=y_2$. Thus, assume $y_1<z<y_2$. We get $y_2-z=z-y_1$, or equivalently, $y_1+y_2=2z$. Without loss of generality assume $b>z$ (the case $b<z$ is analogous). We have $d(v_1,(a,b))=|x_1-a|+b-y_1$ and $d(v_2,(a,b))=|x_2-a|+|b-y_2|$. Since $x_1=x_2$, it is sufficient to prove $|b-y_2|<b-y_1$. If $b\geq y_2$, we have $|b-y_2|-(b-y_1)=y_1-y_2<0$. Otherwise, $|b-y_2|-(b-y_1)=y_1+y_2-2b=2(z-b)<0$.
\end{proof}

As mentioned above, the following was proved in \cite{ACM14}.

\begin{proposition}
Any minimal landmark set consisting of exactly three vertices has
one of the following forms (up to rotating the grid or mirroring
the grid).
\begin{itemize}
\item $L=\{(1,z),(m,z'),(1,z_1)\}$, where $1 < z<z'<z_1\leq m$ or
$1\leq z<z'<z_1<m$ (that is, at most one of $(1,z)$ and $(m,z')$ is a corner).
\item $L=\{(1,z),(m,z'),(m,z_2)\}$, where $1\leq z_2 <z<z'\leq m$ or
$1\leq z<z'<z''<m$.
\item $L=\{(1,z),(m,z),(y,z_3)\}$, where
$1<z<m$ and $z_3 \neq z$.
\end{itemize}
\end{proposition}
\begin{proof}
Consider a minimal landmark set $L$ of cardinality $3$. By our assumption, and landmark set contains $(1,z)$ and $(m,z')$, and we analyze the options for the third vertex of the landmark set. If $z=z'$, by Lemma \ref{zzprime} the third vertex can be any vertex whose second component is not $z$, and the landmark set is of the third type.

Otherwise, by Lemma \ref{lemab}, $L$ contains a vertex $(a,b)$ where either $a=1$ and $b > z'$ hold or $a>1$ and $b \leq z$ hold. Since $|L|=3$, this is the third vertex of $L$.

In the first case let $z_1=b$. It does not hold that both $z=1$ and $z_1=m$, as in this case $\{(1,z),(1,z_1)\}$ is a landmark set, so the set $L$ would not be minimal. The resulting form of $L$ is of the first kind.

Otherwise, let $b \leq z$. Since the landmark set only has one additional vertex (except for $(1,z)$ and $(m,z')$), by applying the same property of Lemma \ref{ab} and rotating the grid, we find that if $b\leq z'$, then $a=m$, and the structure of the landmark set is of the second kind (and we let $y=a$ and $z_3=b$).
\end{proof}

In the case of minimal landmark sets with at least four vertices, we will assume $z<z'$ due to the following. Consider a minimal landmark set (which has the elements $(1,z)$ and $(m,z')$). If $z=z'$, by Lemma \ref{threetogether}, the landmark set has no other vertex of the same column. Assume that it has at least two additional vertices. Then, by Lemma \ref{zzprime}, one vertex can be removed, such that the remaining set is a landmark set.

In order to define an algorithm for finding a minimum weight
landmark set among minimal landmark sets of cardinality at least
$4$, we define a concept called {\it zigzag} sequence. This is a
sequence of an even number (at least four) of vertices,
$q_1,q_2,\ldots,q_{2k}$ for $k \geq 2$, where $q_i=(s_i,d_i)$
satisfies the following properties. First, $s_1=1$ and $s_{2k}=m$
(and it will follow from the definition that $d_1<d_{2k}$). For
even values of $i$ ($i=2,4,\ldots,2k$), $d_i=d_{i-1}$ and
$s_i>s_{i-1}$ hold, and for odd values of $i$
($i=3,5,\ldots,2k-1$), $s_i=s_{i-1}$ and $d_i>d_{i-1}$ hold. That
is, a zigzag sequence starts in the first row, in even steps the
next vertex is below the previous vertex (in the same column), and
in odd steps, the next vertex is to the right of the previous
vertex (and in the same row). The last vertex is in the last row.

Given a zigzag sequence $q_1,q_2,\ldots,q_{2k}$, we say that a
sequence $t_1,t_2,\ldots,t_{2k}$ where $t_i=(b_i,c_i)$ corresponds
to this zigzag sequence (or it is a corresponding sequence) if
$t_1=q_1$, for even values of $i$, $b_i=s_i$ and $r_i \leq d_i$,
and for odd values of $i$ ($i>1$), $r_i = d_i$ and $b_i  \leq
s_i$. Additionally, $r_{2k}>c_1$. That is, the first vertex is the
same in both sequences. In even steps, the vertex of the
corresponding sequence is to the left of the vertex of the zigzag
sequence (in the same row, and they can possibly be equal), and in
odd steps, the vertex of the corresponding sequence is above the
vertex of the zigzag sequence (in the same column, and they can
possibly be equal). The last vertex is in the last row, and its
column must be larger of that of the first vertex.

A sequence that
corresponds to a zigzag sequence $S$ is called a {\it perfect} sequence (for this zigzag sequence) if it
is the minimum cost sequence that corresponds to $S$. Since the
condition for every $i$ such that $1 \leq i \leq 2k$ where $2k$ is
the length of $S$ (the condition on which vertex can be the $i$th vertex of the corresponding sequence)
is independent of other values of $i$. To obtain
a perfect sequence that corresponds to $S$, it is required
to select for each $i$ a vertex of minimum cost that satisfies the
condition of a corresponding sequence. That is, for $i=1$ there is a unique vertex that can be the first vertex of the corresponding sequence, for an even step, a minimum cost vertex whose row is the same as the $i$th vertex of the zigzag sequence and its column is no larger than the column of the $i$th vertex of the zigzag sequence, for an odd step, a minimum cost vertex whose column is the same as the $i$th vertex of the zigzag sequence and its row is no larger than the row of the $i$th vertex of the zigzag sequence, and if $i=2k$, the last vertex of the corresponding sequence has also a restriction on its column, that it is larger than the column of the first vertex of the zigzag sequence (and the corresponding sequence). Note that the conditions on the vertices of the corresponding sequence are independent of each other, and each of the $2k$ vertices has a separate condition.

We note that while a sequence corresponding to a zigzag sequence
defines its zigzag sequence in a unique way given a specific
orientation of the grid, if we rotate the grid (by 180 degrees),
and use the same sequence, the zigzag sequence may be different.

The following theorem connects zigzag sequence and
landmark sets. Figure \ref{zigzag} illustrates this idea.
\begin{figure} [h!]
\hspace{0.8in}
\includegraphics[angle=0,width=0.7\textwidth]{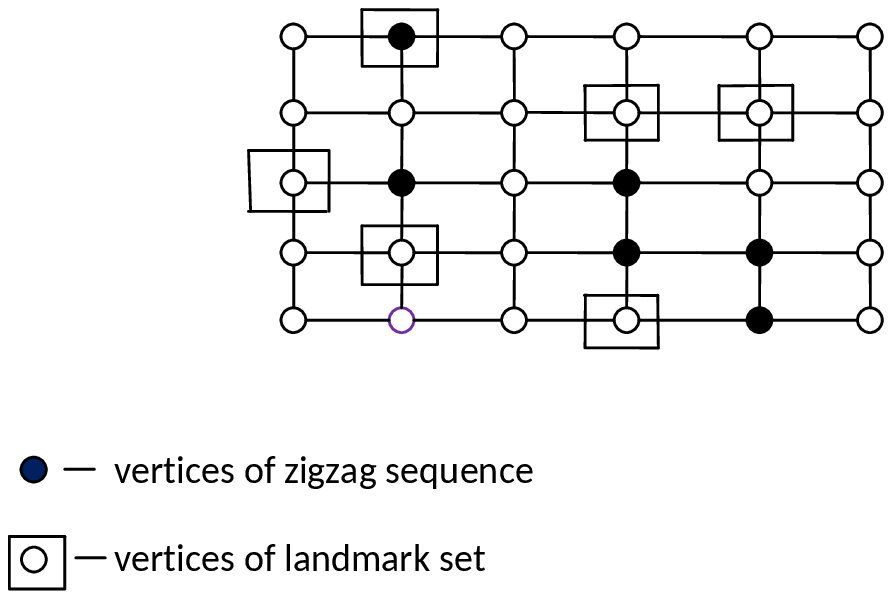}
\caption{A zigzag sequence and a corresponding sequence (which is a landmark set). Notice that the vertex on row $1$ is a common vertex of the two sequences. \label{zigzag}} 
\end{figure}

\begin{theorem}
Every sequence that corresponds to some zigzag sequence is a
landmark set. The vertices of every landmark set that has
cardinality at least $4$, and it is both a minimal landmark set
and  a minimum landmark set, can be ordered to form a perfect
sequence for some zigzag sequence.
\end{theorem}
\begin{proof}
We start with the first property, that is, we show that a sequence
corresponding to a zigzag sequence is a landmark set. A sequence
$t_1,\ldots,t_{2k}$ that corresponds to a zigzag sequence
$q_1,q_2,\ldots,q_{2k}$ (where $t_i=(b_i,c_i)$ and
$q_i=(s_i,d_i)$) satisfies $t_1=q_1=(1,d_1)$, $r_{2k}>c_1=d_1$,
and $b_{2k}=s_{2k}=m$. Thus, by Lemma \ref{onlydiag}, it is
sufficient to consider a pair of vertices of the sub-grid of $t_1$
and $t_{2k}$ that are on a joint diagonal. Since this sub-grid is
contained in the sub-grid of $q_1$ and $(m,n)$, we will prove the
condition on separation of pairs of vertices of a sub-grid that
are joint diagonals for the sub-grid of $q_1=t_1$ and $(m,n)$. We
will show the following by induction. The vertices of the prefix
$t_1,\ldots,t_i$ separate any pair of vertices of the sub-grid of
$q_1$ and $(m,n)$ that are on a joint diagonal, possibly excluding
pairs of vertices of the sub-grid of the vertices $q_i$ and
$(m,n)$ that are on a joint diagonal. Since the sub-grid of
$q_{2k}=(m,d_{2k})$ and $(m,n)$ has no such pairs (every diagonal
of has at most one vertex of this sub-grid), the claim will
follow. The base of the induction is trivial (as the claim is
empty for this case). Assume that the requirements (which we are
proving by induction) hold for a given value $i$, where $i<2k$.
The vertex $t_{i+1}$ satisfies the conditions of Lemma
\ref{cutting} with respect to the sub-grid of $q_i$ and $(m,n)$.
Thus, the induction step follows from Lemma \ref{cutting}.

Next, we consider the second property. We will prove that any
minimal landmark set contains a subset, such that this subset can
be sorted into a sequence corresponding to some zigzag sequence.
Since the set of elements of a sequence corresponding to a zigzag
sequence was proved to be a landmark set, this shows that the
selected subset cannot be a proper subset of the landmark set (as we are
already considering a minimal landmark set). Thus, this will prove that
any minimal landmark set can be sorted to form a sequence that
corresponds to a zigzag sequence. It is also required to show that if
the landmark set is not only a minimal landmark set but it is also
a minimum landmark set, then its sorted sequence is a perfect
sequence that corresponds to the zigzag sequence. Consider a
minimal landmark set that its sorted sequence is not a perfect
sequence for the zigzag sequence. As it is not a perfect sequence
for the zigzag sequence, at least one landmark can be replaced
such that the resulting sequence still corresponds to the same
zigzag sequence but it has a smaller cost. This is possible as the vertices of the corresponding sequence can be selected independently of each other.
Since a landmark set of a smaller cost exists, we find that the considered
landmark set is not a minimum landmark set. Therefore, to complete
the proof of the second property, it remains to show how a subset
of any minimal landmark set of cardinality at least $4$ can be
selected and ordered such that a zigzag sequence can be defined
for it (where the subset of the minimal landmark set will
correspond to this zigzag sequence). This will hold in particular for a minimal landmark set that is also minimum, in which case the corresponding sequence will be perfect. We will
use the notation as in the definition of a zigzag sequence and a
corresponding sequence.

Recall that we assume that $(1,z)$ and $(m,z')$ belong to the
landmark set, such that $z<z'$, and no vertex $(1,\tilde{z})$ with
$z<\tilde{z}\leq z'$ belongs to the set, and no vertex
$(m,\hat{z})$ with $z \leq \hat{z} < z'$ belongs to the set. We
define the zigzag sequence and its corresponding sequence
inductively, such that the last vertex of the corresponding
sequence is $(m,z')$ (that is, the selection process of vertices
from the landmark set ends when this vertex is selected).

Let $q_1=t_1=(1,z)$, that is, the first vertex of both sequences is fixed.
In an odd step of index $i \geq 3$, given $q_{i-1}=(s_{i-1},d_{i-1})$, we will select a vertex of the landmark set that was not selected yet to be $t_i$, such that the the first component of the vertex is no larger than $s_{i-1}$, and its second component is above $d_{i-1}$. In an odd step of index $i$, given $q_{i-1}=(s_{i-1},d_{i-1})$, we will select a vertex of the landmark set that was not selected yet to be $t_i$, such that the the first component of the vertex is no larger than $s_{i-1}$, and its second component is above $d_{i-1}$. In a case of ties, we will select a vertex whose second component is maximum. That is, we select a vertex $t_i$ that is above the sub-grid of the vertices $(s_{i-1},d_{i-1})$ and $(m,n)$. In this case, $q_i=(s_{i-1},c_i)$.
In an even step of index $i \geq 2$, given $q_{i-1}=(s_{i-1},d_{i-1})$, we will select a vertex of the landmark set that was not selected yet to be $t_i$, such that the the second component of the vertex is no larger than $d_{i-1}$, and its first component is above $s_{i-1}$. In a case of ties, we will select a vertex whose first component is maximum. Moreover, if $(m,z')$ is a valid candidate, $t_i$ is defined to be $(m,z')$ (this does not contradict the tie breaking rule). If $q_i$ is defined such that $s_i=m$, the process is terminated.
That is, we select a vertex that is to the left of the sub-grid of the vertices $(s_{i-1},d_{i-1})$ and $(m,n)$. In this case, $q_i=(b_i,d_{i-1})$.
We will show by induction that this is always possible, that is, such a vertex always exist, and that the process terminates. If $t_i$ and $q_i$ are defined in every step, the process terminates, and $t_{2k}=(m,z')$, the sequences satisfy the requirements of a zigzag sequence and its corresponding sequence. We will prove a number of properties by induction, and it will follow from the proof that the sequences were defined properly.

More precisely, we prove by induction that the following properties hold after defining $t_i$ and $q_i$.
\begin{enumerate}
\item The vertices $t_i$ and $q_i$ are well-defined.
\item For any $i'$ such that $1\leq i'\leq i$, it holds that $b_{i'} \leq s_{i}$ and $c_{i'} \leq d_i$.
\item If $i$ is even, then the landmark set has no vertex $(x,y)$ such that $x>s_i$ and $y \leq d_i$.
\item If $i$ is odd, then the landmark set has no vertex $(x,y)$ such that $x \leq s_i$ and $y > d_i$.
\item If $i$ is even then either $s_i \leq m-1$ or $t_i=(m,z')$. Moreover, if $s_i \leq m-1$, then $d_i <z'$.
\item If $i$ is odd, then $s_i \leq m-1$.
\end{enumerate}

Consider the case $i=1$. The first property holds as we defined $t_1$ and $q_1$. The second property holds as the only relevant value of $i'$ is $1$, and by $t_1=q_1$. Since $i$ is odd, we prove the fourth and sixth properties. The sixth property holds as $s_1=1$.  Recall that the landmark set is minimal and its cardinality is at least $4$. To prove the fourth property, we show that the landmark set has no vertex  $(1,\bar{z})$, where $z<\bar{z}\leq n$. By the choice of
$(1,z)$ and $(m,z')$, the landmark set has no vertex
$(1,\bar{z})$, where $z<\bar{z}\leq z'$. If the landmark set has a
vertex $(1,\bar{z})$, where $z'<\bar{z}\leq n$, we can prove that the
landmark set is not minimal. If $z=1$ an $\bar{z}=n$ both hold,
then $\{(1,z),(m,z')\}$ is a set of two adjacent corners and thus
it is a landmark set. Otherwise, $\{(1,z),(m,z'),(1,\bar{z})\}$ is
a landmark set.

Next, consider an even value of $i$. By the induction hypothesis (the sixth property), since the process did not terminate, no vertex of the landmark set of the last row was selected and in particular, $s_{i-1}\leq m-1$. If $d_{i-1}=n$, then the vertex $(m,z')$ is defined to be $t_i$ as $m>s_{i-1}$ and $z' \leq n$. The vertex $q_i$ is defined as $(m,n)$. The first property is satisfied, and all remaining properties hold trivially. Otherwise, assume $d_{i-1} \leq n-1$.
By the induction hypothesis, all landmarks $t_1,t_2,\ldots,t_{i-1}$ have first components in $[1,s_{i-1}]$ and second components in $[1,d_{i-1}]$. Thus, their shortest paths to the vertices $(s_{i-1}+1,d_{i-1})$ and $(s_{i-1},d_{i-1}+1)$ traverse $(s_{i-1},d_{i-1})$, and none of $t_1,t_2,\ldots,t_{i-1}$ separates them. The landmark set has at least one vertex separating them. Moreover, any vertex $(x,y)$ where $x\geq s_{i-1}+1$ and $y\geq d_{i-1}+1$ has shortest paths to these two vertices traversing $(s_{i-1}+1,d_{i-1}+1)$.
By the induction hypothesis, the landmark set has no vertex $(x,y)$ where $x \leq s_{i-1}$ and $y \geq d_{i-1} +1$, and therefore it has at least one vertex $(x,y)$ where $x \geq s_{i-1}+1$ and $y \leq d_{i-1}$. Such a vertex is selected as $t_i$, and $q_i$ is defined such that the requirements of a zigzag sequence and its corresponding sequence are satisfied to be $q_i=(b_i,d_{i-1})$. Since $s_i>s_{i-1}$ and $d_i=d_{i-1}$, $b_{i'} \leq s_{i}$ and $c_{i'} \leq d_i$ holds for any $i'<i$ using the induction hypothesis. Moreover, $b_{i} = s_{i}$ and $c_{i} \leq d_{i-1}=d_i$ holds by definition. Thus, the second property holds. If $s_i=b_i=m$, then the third property holds trivially. Otherwise, since $t_i$ was selected to have a maximum first component, and therefore is a vertex with a first component above $s_i$ and second component of at most $d_{i-1}=d_i$ would have been chosen instead $t_i$, if it existed.
If $d_{i-1} \geq z'$, then $(m,z')$ is selected as $t_i$, since $m \geq s_{i-1}+1$ and $z' \leq d_{i-1}$.
If indeed $(m,z')$ is selected, we have $q_i=(m,d_{i-1})$. In this last case the fifth property holds. Assume that $s_i=b_i=m$ while $c_i \neq z'$. This means that $c_i<z'$ and $d_{i-1}<z'$, as otherwise $(m,z')$ could be selected. Since the landmark set has no vertex $(m,\bar{z})$ with $z\leq \bar{z} < z'$, we have $t_i=(m,c_i)$, where $c_i<z$. However, in this case $\{(1,z),(m,z'),(m,\bar{z})\}$ is
a landmark set, contradicting the assumption that the landmark set is minimal and its cardinality is at least $4$. Thus, we are left with the case $s_i=b_i\leq m-1$. If $d_{i-1} \geq z'$, then $(m,z')$ would be a candidate for selection as $t_i$, and thus $d_{i-1} <z'$.

Finally, consider an odd value of $i$. By the induction hypothesis, $s_{i-1} \leq m-1$ and $d_{i-1} \leq z' -1 \leq n-1$. The property that $t_i$ and $q_i$ are well-defined are proved analogously to the case of even $i$, as the landmark set has a vertex separating $(s_{i-1}+1,d_{i-1})$ and $(s_{i-1},d_{i-1}+1)$. The second property holds similarly to the case of even $i$, the fourth property again holds due to the selection rule of a vertex with a maximum second component, and the sixth property holds as $s_i=s_{i-1} \leq m-1$.
\end{proof}

We show how the known result for the cardinality of a minimal
landmark set \cite{ACM14}, where this cardinality is
$\min\{2n-2,2m-2\}$ for $\min\{m,n\}\geq 3$, follows from the
relation to zigzag sequences. In a zigzag sequence, any row has at
most two vertices, while the first row and the last row have one
vertex each.  By definition, a zigzag sequence has at most $2m-2$
vertices. This implies an upper bound of $2m-2$ on the cardinality
of a zigzag sequence, since the number of vertices in the zigzag
sequence and the corresponding landmark set are equal, and we can
consider cardinalities of landmark sets. If $m \leq n$, we are
done. Otherwise, note that a zigzag sequence has at most two
vertices in each column. If it has no vertices of the last column,
it has at most $2n-2$ vertices. Otherwise, by the definition of a
zigzag sequence (where is particular, it has an even number of
vertices), the sequences has exactly two vertices in each column,
including the first column and the last column. Since it has a
vertex of the last row and a vertex of the first row, it has two
opposite corners, contradicting Lemma \ref{opC}. Since zigzag
sequences have even cardinalities, we find that minimal landmark
sets also have even cardinalities, except for those that have
cardinality $3$.

The action of the algorithm starts with computing the following
values, which are prefix, suffix, and range minima of rows and
columns of the grid. For $1 \leq i \leq m$ and $1 \leq j \leq n$,
let $$Rpref_i^j=\min_{1 \leq k \leq j} c_{i,k} \mbox{ \ \ , \ \ }
Rsuff_i^j=\min_{j \leq k \leq n} c_{i,k} \ , $$
$$Cpref_j^i=\min_{1 \leq k \leq i} c_{k,j} \mbox{ \ \ , \ \ } Csuff_j^i=\min_{i
\leq k \leq n} c_{k,j} \ . $$ These are the prefix minima of rows,
suffix minima of rows, prefix minima of columns, and suffix minima
of columns, respectively. All the $Rpref_{i,j}$ values for a given
value of $i$ can be computed together in time $O(n)$ by the
following simple dynamic programming formulation:
$Rpref_i^1=c_{i,1}$, and for $j \geq 2$,
$Rpref_i^j=\min\{c_{i,j},RL_i^{j-1}\}$. Similarly, all $2m+2n$
values can be computed in time $O(m+n)$. For the side row and
columns we also compute range minima. For $i=1,m$ and any $1 \leq
j \leq \ell \leq n$, let $Rrange_i^{j,\ell}=\min_{j \leq k \leq
\ell} c_{i,k}$, and for $j=1,n$ and any $1 \leq i \leq t \leq m$,
let $Crange_j^{i,t}=\min_{i \leq k \leq t} c_{k,j}$. For a given
value of $j$, all values  $Rrange_1^{j,\ell}$ and
$Rrange_m^{j,\ell}$ can be computed in time $O(n)$, while
$Crange_1^{i,t}$ and $Crange_n^{i,t}$ can be computed in time
$O(m)$ for a given value of $i$. Thus, the time of computing all
the values $Rrange_1^{j,\ell}$, $Rrange_m^{j,\ell}$,
$Crange_1^{i,t}$, and $Crange_n^{i,t}$ is $\Theta(m^2+n^2)=O(|V|^2)$.
It is possible to keep also the identities of the vertices of
minimum costs using the same running time.

Our algorithm computes candidate landmark sets and selects a set
of minimum weight among these sets. There are four landmark sets
of cardinality $2$, and they can be enumerated in time $O(1)$.
There are two types of landmark sets of cardinality three. The
first kind is where two landmarks are on opposite sizes, sharing
the same row or column (depending on which sides these are), and a
third landmark can be any vertex not on the same two or column as
the two other landmarks. There are $O(m)$  candidate pairs on rows
and $O(n)$ pairs on columns. Using the values defined above (the
values $Rpref_i^n=Rsuff_i^1$ and $Cpref_j^m = Csuff_j^1$), we can
find a vertex of minimum cost in the grid $(x_m,y_m)$, another
vertex that has minimum cost out of vertices on other columns (not
on column $y_m$), and another vertex that has minimum cost out of
vertices on other rows (not on row $x_m$). The last two vertices
are distinct from $(x_m,y_m)$, but both of them can possibly be
the same vertex. These two or three vertices can be computed in
time $O(m+n)$. For each pair on a column, and given the (at most)
three vertices defined here, we find a minimum cost vertex that is
not on a certain column or not on a certain row in time $O(1)$ for
each candidate pair. The second kind of landmark sets with
cardinality three consists of two vertices on one side, and one
vertex on the opposite side, on a column that is strictly between
the columns of the first two vertices. For every vertex $v$ on a
side, it is possible to use the values $Rrange_1^{j,\ell}$,
$Rrange_m^{j,\ell}$,  $Crange_1^{i,t}$, and $Crange_n^{i,t}$ to
find two vertices of minimum cost on the opposite side, such that
the vertex is between them. For $v=(1,z)$, such that $1<z<n$, the
costs of the required two vertices are $Rrange_m^{1,z-1}$ and
$Rrange_m^{z+1,n}$. For $v=(m,z)$, such that $1<z<n$, the costs of
the required two vertices are $Rrange_1^{1,z-1}$ and
$Rrange_1^{z+1,n}$. For $v=(q,1)$, such that $1<q<m$, the costs of
the required two vertices are $Crange_n^{1,q-1}$ and
$Crange_n^{q+1,m}$. For $v=(q,n)$, such that $1<q<m$, the costs of
the required two vertices are $Crange_1^{1,q-1}$ and
$Crange_1^{q+1,m}$. Thus, a landmark set of minimum weight of this
form can be found in time $O(m+n)$, as the candidates for the
vertex $v$ are all side vertices (excluding corners).

Finally, in order to find a minimum weight landmark set out of
landmark sets of cardinality at least $4$, we define a dynamic
programming algorithm. We present the algorithm for the case where
the landmark set corresponds to a zigzag sequence as defined
above. In order to consider all relevant subsets, the algorithm is
applied eight times, such that the grid is rotated and mirrored in
all possible directions. We define the following functions for any
vertex $v=(1,z)$ with $1\leq z\leq n-1$. For any vertex $u=(a,b)$,
let $F_v^o(a,b)$ denote the minimum weight of an odd length prefix
of a sequence that corresponds to a prefix of a zigzag sequence
whose first vertex is $v$  and the last vertex is $u$, and let
$F_v^e(a,b)$ denote the minimum weight of an even length prefix of
a sequence that corresponds to a prefix of a zigzag sequence whose
first vertex is $v$ and the last vertex is $u$. We also let
$G_v^o(a,b)=\min_{1 \leq i \leq a} F_v^o(i,b)$ and
$G_v^e(a,b)=\min_{1 \leq j \leq b} F_v^e(a,j)$.

We have $$G_v^o(1,z)=F_v^o(1,z)=c(z) \mbox{ \ \ \ and \ \ \ }
G_v^e(1,z)=F_v^e(1,z)=\infty \ . $$ Moreover, for any $u=(1,z')$
such that $z'\neq z$,
$$G_v^e(1,z')=F_v^e(1,z')=F_v^o(1,z')=G_v^o(1,z')=\infty\ . $$

For $u=(r',z')$, where $1<r'<m$, we let
$$F_v^o(r',z')=G_v^e(r',z'-1)+Cpref_{z'}^{r'},
G_v^o(r',z')=\min\{F_v^o(r',z'),G_v^o(r'-1,z')\} \ , $$
$$F_v^e(r',z')=G_v^o(r'-1,z')+Rpref_{r'}^{z'}, \mbox{\ \ and \  \
} G_v^e(r',z')=\min\{F_v^e(r',z'),G_v^e(r',z'-1)\} \ . $$ Finally,
for $u=(m,z')$, we let  $G_v^o(m,z')=F_v^o(m,z')=\infty$,
$G_v^e(m,z')=F_v^e(m,z')=\infty$, if $z'\leq z$, and if $z'>z$, we
let $$F_v^e(m,z')=G_v^o(m-1,z')+Rrange_m^{z+1,z'} \ , $$ and
$G_v^e(m,z')=\min\{F_v^e(m,z'),G_v^e(m,z'-1)\}$. The minimum cost
of a sequence corresponding to a zigzag sequence starting at $v$
is $G_v^e(m,n)$, and the set of vertices can be found via
traceback. The running time for a fixed vertex $v$ is $O(|V|)$.

The total running time of the algorithm is $\Theta((m+n)|V|)=O(|V|^2)$.

\end{document}